\documentclass[prx,preprintnumbers,twocolumn,superscriptaddress,nofootinbib]{revtex4-2}

\usepackage{ulem} 
\expandafter\let\csname equation*\endcsname\relax
\expandafter\let\csname endequation*\endcsname\relax   
\usepackage{euscript}
\usepackage{hyperref}
\usepackage{mathtools}
\usepackage{cases}
\usepackage{graphicx}
\usepackage{amsmath}
\usepackage{amssymb}
\usepackage{amsthm}
\usepackage{bbm}
\usepackage{microtype}
\usepackage{xurl}
\hypersetup{breaklinks=true}
\newcommand{\beq}{\begin{equation}}
	\newcommand{\eeq}{\end{equation}}
\newcommand{\bqa}{\begin{eqnarray}}
	\newcommand{\eqa}{\end{eqnarray}}

\usepackage[utf8]{inputenc}
\usepackage{amsmath}
\usepackage{amsfonts}
\usepackage{amssymb}
\usepackage{graphicx}
\usepackage{physics}
\usepackage{color}
\usepackage{dsfont}
\usepackage{euscript}

\newcommand{\G}{\mathsf{M}}
\newcommand{\M}{\mathsf{M}}

\newcommand{\T}{\mathsf{T}}

\newcommand{\h}{\mathcal{H}}							

\newcommand{\obs}[1][2]{\bm{\mathbf{Obs}}({1},{2})} 

\newcommand{\tra}[1]{\text{Tr}\left\{ #1\right\}}   

\newtheorem{corollary}{Corollary}
\newtheorem{theorem}{Theorem}

\newtheorem{prop}{Proposition}
\newtheorem{lemma}{Lemma}

\begin{document}
	\title{Multiple-shot labeling of quantum observables}
\author{Seyed Arash Ghoreishi}
\email{arash.ghoreishi@savba.sk}
\affiliation{Institute of Physics, Slovak Academy of Sciences, Dubravska cesta 9, Bratislava 845 11, Slovakia}
\author{Nidhin Sudarsanan Ragini}
\affiliation{Institute of Physics, Slovak Academy of Sciences, Dubravska cesta 9, Bratislava 845 11, Slovakia}
\affiliation{Department of Computer Science, Virginia Tech, Blacksburg, VA 24061, USA}
\affiliation{Virginia Tech Center for Quantum Information Science and Engineering, Blacksburg, VA
24061, USA}

\author{Sk Sazim}
\affiliation{Center for Theoretical Physics, Polish Academy of Sciences, Aleja Lotnik\'{o}w 32/46, 02-668 Warsaw, Poland}
\author{Mario Ziman}
\affiliation{Institute of Physics, Slovak Academy of Sciences, Dubravska cesta 9, Bratislava 845 11, Slovakia}
	
\begin{abstract}
Quantum labeling tasks ask one to recover the missing associations between classical outcome labels and the
effects forming the POVM. We study labeling in the multiple-shot regime, allowing a finite number of uses of the device and the most general
tester-based strategies, including adaptivity. For binary observables, we show that if perfect labeling is
impossible in a single shot, then it remains impossible with any finite
number of shots. In particular, we derive the formula for minimum-error performance and highlight its ``even-odd" behavior. For non-binary observables, we derive the optimal single-shot minimum-error success
probability in closed form, and show that entanglement assistance does not improve this optimum. We also
provide finite-shot schemes for (perfect or partial) labeling and give illustrative examples, including the
qubit trine POVM where the optimal two-shot success probability is computed explicitly.
	\end{abstract}
	
	\maketitle
\section{Introduction}
\noindent Quantum theory furnishes a prescription for describing the probabilistic and statistical aspects of experimental scenarios \cite{holevo2011probabilistic}. In particular, when we wish to predict the observed measurement statistics of an experiment, each outcome is associated with a positive operator, referred to as an effect \cite{heinosaari2011mathematical}. Although an effect does not itself describe the underlying physical mechanism, it determines the probability of occurrence of the associated outcome. To characterise the measurement statistics, we associate a set of effects, together with classical labels, to each mutually exclusive outcome. Moreover, while the effects corresponding to a measurement device can be experimentally verified, the choice of labels cannot be tested using the device alone.

There exist scenarios in which the mathematical description of the effects associated with a measurement device implementing an observable is known in advance, while their pairings with outcome labels are lost or unknown. Such situations may arise due to deficiencies on the part of the users of the device and require the identification of the missing associations, or ``labels,'' in order to correctly assign a mathematical description to each label. Since the performance of measurement is crucial for the successful development of quantum technologies \cite{heinosaari2023anticipative,cieslinski2024analysing,ghoreishi2025future,bozzio2025quantum}, a class of tasks designed to address this effect--outcome association problem was introduced in Ref.~\cite{sudarsanan2024single} and termed quantum labeling tasks. These tasks were identified as a particular instance of quantum distinguishability problems involving observables~\cite{sudarsanan2024single} 
More broadly, quantum distinguishability tasks trace their origins to the seminal work of Helstrom in 1969~\cite{helstrom1969quantum}, where they were first studied in the context of quantum states~\cite{barnett2009quantum, bae2015quantum, bae2013structure, ghoreishi2019parametrization, rouhbakhsh2023geometric,Lue2026,Lue2026a,Ivanovic1987,Peres1988,Dieks1988,Jaeger1995}, and were later extended to quantum measurements and quantum channels~\cite{ziman2009unambiguous, ziman2010single, sedlak2014optimal, chiribella2008memory, bavaresco2021strict, ohst2024characterising}.

 Although Quantum state discrimination can be viewed as a special case of quantum network (or channel) discrimination, obtained by regarding states as preparation channels with trivial input, there are crucial differentiating features when finitely many copies or uses (hereafter referred to as ``shots") of the associated quantum devices are accessible to assist the distinguishability task. The following contrast can be observed. Quantum states that are not perfectly distinguishable in a single shot remain imperfectly distinguishable even when finitely many copies are available. In contrast, for quantum channels, it can happen that two channels are not perfectly distinguishable with a single use, yet become perfectly distinguishable with finitely many uses, provided one is allowed to employ suitable multi-use probing strategies (possibly involving entanglement and/or adaptivity). A well-known example is any pair of distinct unitary channels, which can be perfectly distinguished with some finite number of uses \cite{acin2001statistical,d2001using}. More generally, however, such a finite-use perfect discrimination is not guaranteed for arbitrary channels and occurs only under specific conditions. These observations motivate us to investigate potential performance improvements in labeling tasks in the multiple-shot regime.

In this paper, we extend the single-shot labeling framework of Ref.~\cite{sudarsanan2024single} to the multi-shot setting,
treating both binary and non-binary observables. We derive explicit expressions for the optimal
minimum-error performance and use them to quantify how access to multiple uses changes the
labeling power. For non-binary observables, we show that complete labeling cannot be achieved in
a single use, and we identify conditions under which perfect labeling becomes possible with finitely
many uses. Moreover, since measurement channels are entanglement-breaking, entanglement assistance does not yield an advantage: the relevant optimisations reduce to classical strategies based on choosing probe states (possibly adaptively) and classical post-processing.
In particular, we show that a $d$-dimensional observable with $n$ effects can be perfectly labeled in
$n-1$ shots (via a simple scheme) if and only if at least $n-1$ effects have an eigenvalue equal to $1$.

The paper is organised as follows. In Sec.~\ref{sec:formulation} we review quantum observables and quantum testers and
formulate labeling as a channel-discrimination problem. In Sec.~\ref{sec:3-binary} we analyse binary observables,
establishing impossibility results for finite-shot perfect labeling and characterising the optimal
minimum-error performance. In Sec.~\ref{sec:4non-binary} we turn
to non-binary observables: we characterise finite-shot perfect labeling, study minimum-error labeling,
and discuss multi-shot protocols, including a detailed trine-measurement case study. We conclude in
Sec.~\ref{sec:conclusion} with a summary and outlook.

\section{formulation of the problem} \label{sec:formulation}
\subsection{Quantum observables}

\noindent An $n$-valued quantum measurement is characterised by its mutually exclusive outcomes $x_1, \dots, x_n$, which are respectively associated with effects, $M_1, \dots, M_n$; these positive operators are constrained to sum up to the identity operator, $M_1 + \dots + M_n = \mathbbm{1}$. Even though we can safely identify the measurements as collections of effects for most of the purposes  of theoretical study, they are rigorously identified with maps which assign effects to particular outcomes. If we denote $\Omega := \{x_1,\dots, x_n \}$ as the ordered set of {outcome labels} and $\mathcal{E}(\mathcal{H})$ as the set of all effects, then a quantum measurement is identified with a normalised positive operator-valued measure (POVM), which we refer to as an ``{observable}" $\mathsf{M}: \Omega \rightarrow \mathcal{E}(\mathcal{H})$ with $\mathsf{M}(x_k) = M_k$. Moreover, these observables can be modeled as quantum-to-classical measure-and-prepare channels $\mathcal{M}$ (hereafter referred to as ``{measurement channels}"). These channels take state $\varrho$ of the measured system as input and outputs state of an $n$-dimensional system, with orthonormal basis $\{ \ket{x_k} \}_{k=1}^n$,  as, 
\begin{equation}
    \mathcal{M}(\varrho) = \sum_{x_k \in \Omega} \text{Tr}\{\varrho M_k\} \dyad{x_k}.
    \label{m&p}
\end{equation}
Since we have the channel representation, as a result of the Choi-Jamiołkowski isomorphism \cite{choi1975completely,jamiolkowski1972linear} between completely positive maps and positive operators, we can associate an observable $\mathsf{M}$ to the positive operator $\Phi_{\mathsf{M}}$ given by,
\begin{equation}
    \Phi_{\mathsf{M}} = (\mathbbm{1}\otimes \mathcal{M})[\Psi_+] = \sum_{k} M_k^{\top} \otimes \dyad{x_k},
\end{equation}
where $\Psi_+ = \dyad{\psi_+}$ is the unnormalised maximally entangled state, with $\ket{\psi_+} := \sum_{j=1}^d \ket{jj} \in \h \otimes \h$ and $\{  \ket{1}, \dots, \ket{d} \}$ being an orthonormal basis of $\h$. \par 

We study discrimination scenarios when multiple queries or copies of each of the measurement channels involved are available. These multiple copies, which are generally implemented at different points in time, can be considered as a single multiple-time-step quantum process. Such processes are, in a holistic fashion, mathematically described by ``{quantum combs}", which are nothing but Choi-Jamiołkowski operators of these processes satisfying certain causality constraints \cite{chiribella2008quantum, chiribella2009theoretical}. Instead of digressing towards discussing general quantum combs and their properties, we restrict ourselves to stating that the comb for an $n$-copy measurement channel $\mathcal{M}$ is given by the Choi-Jamiołkowski operator of the comb as $\Phi_{\mathsf{M}}^{\otimes n}$.  Interested readers can visit Ref.\cite{chiribella2009theoretical} for general discussions on combs. 

\subsection{Quantum testers}
\noindent The most general experiments or test procedures ``measuring" quantum combs are described by single descriptors referred to as ``{quantum testers}" \cite{chiribella2008quantum, ziman2008process, chiribella2009theoretical}. Within the premise of quantum distinguishability tasks, the generality of testers accommodates for the so-called ``{adaptive}" strategies as well as entanglement-assisted strategies. A quantum tester $\T$ measuring a quantum comb with $n$ time-steps is characterised by an incomplete network with $n$ open slots (see Fig. \ref{fig:testers}). As such, we refer such a $\T$ as an $n$-slot tester, oftentimes denoting it as $\T^{(n)}$. The input systems of the tester are labeled by even numbers, starting with 0 and the output systems with odd numbers (See (c) of \autoref{fig:testers}). 

Given a set of outcomes $\Delta = \{1, \dots, m \}$ associated with the measurement of an $n$ time-step process $\mathcal{C}$, the corresponding $n$-slot tester $\T^{(n)} $ is a collection of positive operators referred to as ``process effects" $\{ T^{(n)}_1, \dots, T^{(n)}_m \}$, which satisfies the following normalisation conditions, 
\begin{eqnarray}
        \sum_{x \in \Delta} T^{(n)}_x &=&  \Xi^{(n)} \otimes \mathbbm{1}_{2n-1},  \\
        \text{Tr}_{2k-2}\{\Xi^{(k)}\} &=& \Xi^{(k-1)} \otimes \mathbbm{1}_{2k-3}, \quad \forall k \in [2,n],  \\
        \text{Tr}\{\Xi^{(1)}\} &=& 1.
\end{eqnarray}
These normalisation conditions reflect the causal structure of the tester network. The probability of observing an event $\alpha \in \Delta$ is given by the generalised Born rule $p_{\alpha} =  \tra{\Phi_{\mathcal{C}}T_{\alpha}}$, where $\Phi_{\mathcal{C}}$ is the quantum comb of the process. Moreover, we observe from these conditions that a $1$-slot tester $\T^{(1)}$ has the normalisation condition $\sum_x T^{(1)}_x =  \xi \otimes \mathbbm{1}_1 $ for some appropriate state $\xi $  on the input system of the channel \cite{ziman2008process}. We note that such 1-slot testers describe measurements of 1-time step quantum processes. Such testers which test measurement channels can be assumed, without loss of generality, to take the form $T_{\alpha} = \sum_{k} H_k^{(\alpha)} \otimes \dyad{x_k}$ , where $\sum_{\alpha} H_k^{(\alpha)} = \xi$, for all $k$ \cite{sedlak2014optimal}. 

\begin{figure*}
    \centering
    \includegraphics[scale=0.2]{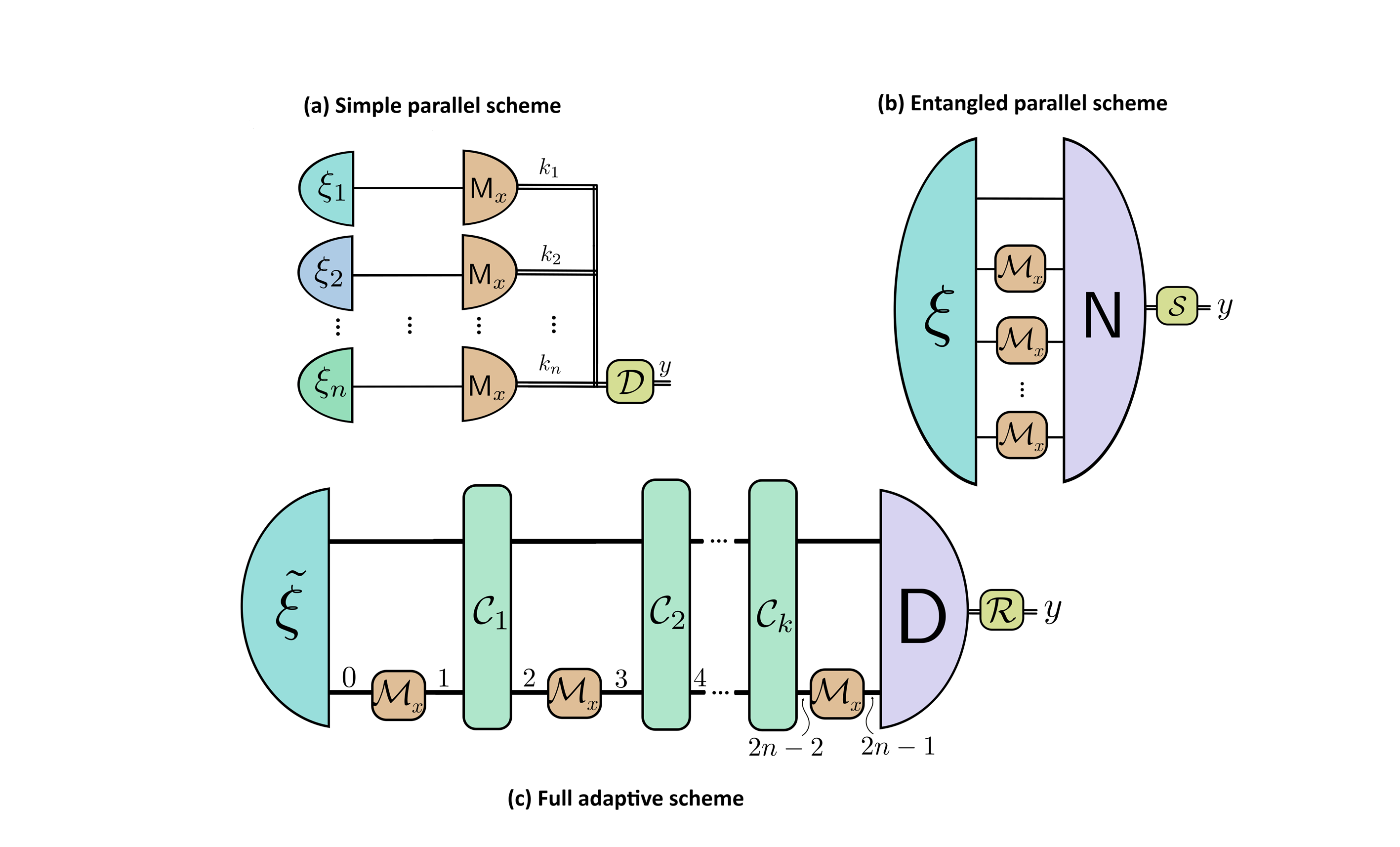}
    \caption{The figure depicts different discrimination tasks using quantum testers.}
    \label{fig:testers}
\end{figure*}
\subsection{Labeling problem}
\label{subsection: SII.Labeling}

\noindent As discussed in Ref. \cite{sudarsanan2024single}, the aspect of labeling an observable arises when there is a lack of information on the mappings between the outcomes and the associated effects, $x_k\mapsto M_k$. For the case where all the effects are different, of an observable with $n$ effects, there are $n!$ different possibilities, due to permutations, of how the outcomes can be paired with the effects. As such, the labeling problem aims to discriminate among these $n!$ possibilities. Once we fix an order, or equivalently the effect-outcome pairings, $\M_{\pi}$ is the observable obtained by $\pi$-permuting the order.

Within our premise of measurement channels, labeling corresponds to discriminating among channels, whose actions are given by, 
\begin{align}
    \mathcal{C}_{\pi}(X) = \sum_k \text{Tr}\{X M_k\} \dyad{x_{\pi(k)}} = \mathcal{P}_{\pi} \circ \mathcal{C}_{\text{id}}(X),
\end{align} 
where $\mathcal{P}_{\pi}(Y) = P_{\pi}Y P_{\pi}$ with $P_{\pi} = \sum_k \dyad{x_{\pi(k)}}{x_k}$ being the permutation operators. Notice that $\mathcal{C}_{\text{id}}$ denote no permutation. We can also observe that this symmetry is translated to the corresponding Choi operators $\Phi_{\pi}$. With $m = n!$, let $T_1, \dots, T_m$ denote the tester describing the labeling experiment associated with $\mathcal{C}_1, \dots, \mathcal{C}_m$. Then, this experiment is characterised by conditional probabilities $p(\pi| \pi') = \text{Tr}\{\Phi_{\pi'} T_{\pi}\}$. The most common performance quantifier of this task is the average error probability, 
\begin{align}
    p_e = 1 -\frac{1}{n!} \sum_{\pi} p(\pi | \pi),
\end{align}
or equivalently, probability of success $p_{s}=1-p_{e}$, where $p(\pi|\pi)$ corresponds to correct decisions. Minimising $p_e$, we arrive at \textit{minimum-error labeling}. Admitting an additional inconclusive outcome, labeled ``?", through $T_?$ and strictly requiring not to admit any error, that is $p(\pi|\pi') = 0 $ for all $\pi \neq \pi'$, we arrive at \textit{unambiguous labeling}. Here, when we enjoy error-free conclusions, it is at a cost of decision failures which occurs with a probability 
\begin{align}
    p_f = \frac{1}{n!} \tra{\sum_{\pi} \Phi_{\pi} T_?}.
\end{align}
Whenever $p_e = p_f =0$ is achieved, we have ``{perfect labeling}".

In multiple-shot labeling experiments, we have access to finite $n$ queries or uses of the measurement channel associated to the unlabeled measurement device. Such experiments are  characterised by conditional probabilities $p(\pi|\pi') = \tra{\Phi_{\pi'}^{\otimes n} T_{\pi}}$. Moreover, we can classify possible strategies as either ``simple parallel", ``entangled parallel tester", or `` full adaptive" schemes (see \autoref{fig:testers}). We note that these strategies are not strictly exclusive in specific instances. Simple parallel schemes do not rely on entanglement: We arrive at conclusions by measuring the unlabeled measurement device on appropriately chosen probe states, then processing the outcomes and deciding. The entangled parallel tester schemes are entanglement-assisted, and the different implementations of the unlabeled observable are not correlated to each other by quantum memories. The full adaptive schemes are not only assisted by entanglement but also by quantum memories and adaptivity. Thus, rendering each implementation of the measurement channel to be correlated to the previous ones in time.

\section{Binary observables}
\label{sec:3-binary}

\noindent We refer to quantum observables composed of two non-zero effects as binary observables. Such an observable, which corresponds to two events, can always be seen as describing some yes-or-no question associated to the measurement \cite{heinosaari2011mathematical}. Labeling an unlabeled binary measurement device corresponds to distinguishing between two binary observables. Suppose these two observables are $\mathsf{M}$, with effects $M_1$ and $M_2$, and $\mathsf{N}$, with effects $N_1$ and $N_2$. Then, from the prior information associated to our labeling problem, we also know that $N_1 = M_2$ and $N_2 = M_1$. Similarly, we denote their corresponding Choi operators as $\Phi_{\mathsf{M}}$ and $\Phi_{\mathsf{N}}$.

\subsection{Perfect labeling}

\noindent Keeping aside the pathological cases of both effects being identical to each other, the following result was found in Ref. \cite{sudarsanan2024single}:  In the single-shot regime, a binary observable can be perfectly labeled if and only if at least one of the two composing effects is not full-rank. Consequently, a natural as well as relevant question to be asked in the multiple-shot regime is: Given a binary observable $\mathsf{M}$ that is not perfectly labelable in a single shot, \text{i.e.,} its effects $M_1$ and $M_2$ are both full rank operators, is it possible to perfectly label it using finite $n$ uses of the unlabeled measurement device, especially when most general strategies are accessible? Investigations into this lead to the following theorem.

\begin{theorem}
    If a binary observable does not admit perfect labeling in a single shot, then it does not admit perfect labeling in any finite number of shots.
\end{theorem}

The theorem says that, regardless of whether we implement parallel strategies or full adaptive strategies, perfect labeling cannot be achieved for the given observables. Consequently, the theorem and its proof can essentially be broken down into the following two propositions and their proofs, respectively.

\begin{prop}
    If a binary observable does not admit perfect labeling in a single shot, then it does not admit perfect labeling in any finite number of shots, when implemented in parallel.
\end{prop}
\begin{proof}
    Let us prove the scenario for two shots of the observable. We can perfectly label the measurement device with parallel implementation of two slots if there exists a  1-slot quantum tester $\mathsf{T} \equiv \{ {T}_{\mathsf{M}}, {T}_{\mathsf{N}}\}$, capable of perfectly distinguishing between the observables $\mathsf{M}$ and $\mathsf{N}$. With the normalisation ${T}_{\mathsf{M}}+{T}_{\mathsf{N}} = \xi \otimes \mathbbm{1}$, for some state $\xi$, the perfect labelability translates to satisfying the condition \cite{chiribella2008memory},
    \begin{equation}
        \Phi_{\mathsf{M}} \otimes \Phi_{\mathsf{M}} (\xi \otimes \mathbbm{1})\Phi_{\mathsf{N}} \otimes \Phi_{\mathsf{N}}=O.
    \end{equation}
    Here, $\Phi_{\mathsf{X}}$ corresponds to Choi-Jamio\l kowski operators of measurement channel of the observable $\mathsf{X}$ and $O$ to the zero operator. Plugging in the expressions for the Choi-Jamio\l kowski operators, the above equation is translated into the following equation, those should be satisfied simultaneously for all $x, y = 1,2$,
    \begin{equation}
        (M_x^{\top} \otimes M_y^{\top}) \xi (N_x^{\top} \otimes N_y^{\top}) = O.
    \end{equation}
    From our assumption that all the involved effects are full rank, their inverses exist. This enables us to sandwich the above equation, left with $(M_x^{\top})^{-1} \otimes (M_y^{\top})^{-1}$ and right with $(N_x^{\top})^{-1} \otimes (N_y^{\top})^{-1}$, leaving us with $\xi = O$. This implies that there does not exist a valid state $\xi$ and consequently a tester $\mathsf{T}$ that would achieve our desired task. This completes the proof for the case of two shots. The same reasoning holds for any finite number of shots.
\end{proof}

\begin{prop}
    If a binary observable does not admit perfect labeling in a single shot, then it does not admit perfect labeling in any finite number of shots, when implemented using adaptive strategies.
\end{prop}
\begin{proof}
    Let us prove the scenario for two shots of the observable. In the case of labeling with the more general adaptive strategies, perfect labeling of $\mathsf{M}$ translates to the existence of binary 2-slot testers $\mathsf{T}^{(2)} \equiv \{ {T}^{(2)}_{\mathsf{M}}, {T}^{(2)}_{\mathsf{N}}\}$, with normalisation  conditions, ${T}^{(2)}_{\mathsf{M}}+{T}^{(2)}_{\mathsf{N}} = \Xi^{(2)} \otimes \mathbbm{1}_3$, $\text{Tr}\{ \Xi^{(2)}\} = \xi \otimes \mathbbm{1}_1$, satisfying the condition
    \begin{equation}
        \Phi_{\mathsf{M}} \otimes \Phi_{\mathsf{M}} (\Xi^{(2)} \otimes \mathbbm{1}_3)\Phi_{\mathsf{N}} \otimes \Phi_{\mathsf{N}}=O,
    \end{equation}
    where $O$ is the null-operator and $\Phi_{\mathsf{M}/\mathsf{N}}$ is the Choi matrix for the measurement channel induced by observable $\mathsf{M}/\mathsf{N}$ \cite{chiribella2008memory}.
    Plugging in the expressions for the Choi-Jamio{\l}kowski operators, we arrive at the equation 
    \begin{equation}
        \sum_{j,k} (M_i^{\top} \otimes M_j^{\top} \otimes\dyad{j}) \Xi^{(2)} (N_i^{\top} \otimes N_k^{\top} \otimes\dyad{k})=O,
    \end{equation}
    
    \noindent which should be satisfied for all $i = 1,2$. Now, sandwiching this equation, with $\mathbbm{1} \otimes (M_i^{\top} )^{-1} \otimes \mathbbm{1}$ from the left and with $\mathbbm{1} \otimes (N_i^{\top} )^{-1} \otimes \mathbbm{1}$ from right and tracing out the system $\h_2$ from this equation, we are left with
    \begin{align}
       0&= \sum_{j,k} (M_j^{\top} \otimes \dyad{j})(\text{Tr}_2\{\Xi^{(2)}\})(N_k^{\top} \otimes \dyad{k}) \nonumber \\
        &=\sum_{j,k} (M_j^{\top} \otimes \dyad{j})(\xi \otimes \mathbbm{1})(N_k^{\top} \otimes \dyad{k}) \nonumber \\
        &=\sum_j (M_j^{\top} \xi N_j^{\top}) \otimes \dyad{j}.
    \end{align}    
    As a result, the above equations of perfect labelability translate to $M_j^{\top} \xi N_j^{\top} = O$ for $\forall j$. Using the same arguments from the previous proof, we can conclude that $\xi$ has to be the null operator, which means a contradiction. Consequently, there does not exist a 2-slot tester that describes full adaptive scheme capable of achieving perfect labelability.The same reasoning holds for any finite $N$ shots, proving the impossibility of perfect labeling for the most general full adaptive scheme.
\end{proof}

\subsection{Minimum-error labeling}
 Since we have seen that binary observables that cannot be labeled in a single shot can neither be labeled in any finite shots, we proceed to evaluate the optimal averaged error probability for labeling, given such observables. This is carried out to compare the error probability obtained with multiple uses with that obtained with a single use, and in general, to see whether the average performance of labeling experiments increases with the number of uses.
 
Let us investigate this scenario for the case of two uses. The averaged error probability that is associated with each decision is given by

\begin{align}
    &p_e = \frac{1}{2} \left\{  \tra{T_{\mathsf{M}} \Phi^{\otimes2}_{\mathsf{N}} } + \tra{T_{\mathsf{N}} \Phi^{\otimes2}_{\mathsf{M}} }  \right\}, 
    \label{eq:binaryminimum}
\end{align}

Using the tester normalization, we can arrive at the following equation

\begin{align}
p_e
&=\frac12\Big(\Tr[T_{\mathsf{M}}\,\Phi_{\mathsf{N}}^{\otimes 2}]
+\Tr[(\Xi^{(2)}\otimes\mathbbm{1}-T_{\mathsf{M}})\,\Phi_{\mathsf{M}}^{\otimes 2}]\Big)\nonumber\\
&=\frac12\Big(\Tr[(\Xi^{(2)}\otimes\mathbbm{1}_{\mathrm{out}})\,\Phi_{\mathsf M}^{\otimes 2}]
+\Tr[T_{\mathsf M}(\Phi_{\mathsf N}^{\otimes 2}-\Phi_{\mathsf M}^{\otimes 2})]\Big).
\end{align}
Since $\Xi^{(2)}$ is a valid two-slot tester normalisation, one has
$\Tr[(\Xi^{(2)}\otimes\mathbbm{1}_{\mathrm{out}})\,\Phi_{\mathsf M}^{\otimes 2}]=1$,
and the optimal strategy can be found by minimizing the error over all testers, which can be written as 
\begin{equation}\label{eq:pe-helstrom-comb}
p_e=\frac12\left(1+\min_{0\le T_{\mathsf M}\le \Xi^{(2)}\otimes\mathbbm{1}}
\Tr\!\left[T_{\mathsf M}\big(\Phi_{\mathsf{N}}^{\otimes 2}-\Phi_{\mathsf{M}}^{\otimes 2}\big)\right]\right).
\end{equation}

For the binary labeling problem 
\begin{equation}\label{eq:swap}
\mathsf{N}_1=\mathsf{M}_2,\qquad \mathsf{N}_2=\mathsf{M}_1,
\end{equation}
So that, for any input state $\rho$, if we let $\Pr_{\mathsf M}(1\mid\rho)=\Tr[\rho,\mathsf M_1]$ denote the probability of obtaining outcome $1$ when the probe state $\rho$ is measured by $\mathsf M$, then the single-use outcome probabilities satisfy
$\Pr_{\mathsf N}(1|\rho)=1-\Pr_{\mathsf M}(1|\rho)$.
Because the channel outputs are classical, any adaptive two-slot tester is equivalent to
a sequential strategy: choose a first input $\rho$, observe the first label $i\in\{1,2\}$,
then choose the second input $\rho_i$ and observe the second label $j\in\{1,2\}$.
Let
\begin{eqnarray}
p:&=\Pr_{\mathsf{M}}(1|\rho)=\Tr[\rho\,\mathsf M_1],\\
q_1:&=\Pr_{\mathsf{M}}(1|\rho_1)=\Tr[\rho_1\,\mathsf M_1],\\
q_2:&=\Pr_{\mathsf{M}}(1|\rho_2)=\Tr[\rho_2\,\mathsf M_1].
\end{eqnarray}
Then the joint distributions under hypotheses $\mathsf{M}$ and $\mathsf{N}$ are
\begin{align}
&P_{\mathsf M}(1,1)=p q_1,\quad P_{\mathsf M}(1,2)=p(1-q_1),\\
&P_{\mathsf M}(2,1)=(1-p)q_2,\quad P_{\mathsf M}(2,2)=(1-p)(1-q_2), \nonumber\\
&P_{\mathsf N}(1,1)=(1-p)(1-q_1),\quad P_{\mathsf N}(1,2)=(1-p)q_1, \nonumber \\
&P_{\mathsf N}(2,1)=p(1-q_2),\quad P_{\mathsf N}(2,2)=p q_2, \nonumber
\end{align}
After fixing a two-shot strategy -- i.e., the first input state $\rho$ and the conditional second input states
$\rho_1,\rho_2$ --the experiment produces a classical result $y=(i,j)\in\{1,2\}^2$.
Under the two hypotheses $\mathsf M$ and $\mathsf N$ this outcome is distributed according to
$P_{\mathsf M}(i,j)$ and $P_{\mathsf N}(i,j)$, respectively.
For equal priors $\Pr(\mathsf M)=\Pr(\mathsf N)=\tfrac12$, the optimal decision rule is the
maximum-a-posteriori rule, which here reduces to comparing likelihoods:
upon observing $(i,j)$ one guesses $\mathsf M$ if $P_{\mathsf M}(i,j)\ge P_{\mathsf N}(i,j)$
(and $\mathsf N$ otherwise).
Hence, conditioned on observing $(i,j)$, the probability of error is the smaller of the two
weighted likelihoods, i.e.\ $\min\{\tfrac12 P_{\mathsf M}(i,j),\tfrac12 P_{\mathsf N}(i,j)\}
$.
Summing over all outcomes gives the average minimum-error probability
\begin{equation}\label{eq:pe-classical}
p_e(p,q_1,q_2)=\frac12\sum_{i,j\in\{1,2\}}\min\{P_{\mathsf M}(i,j),P_{\mathsf N}(i,j)\}.
\end{equation}

If we choose the same input state in the second use, i.e. $q_1=q_2=p$, we have 
\begin{eqnarray}
\sum_{j}\min\{P_{\mathsf M}(1,j),P_{\mathsf N}(1,j)\}=\min\{p,1-p\},\\
\sum_{j}\min\{P_{\mathsf M}(2,j),P_{\mathsf N}(2,j)\}=\min\{p,1-p\},
\end{eqnarray}
as a result
\begin{equation}\label{eq:pe-two-use}
p_e(p,q_1{=}p,q_2{=}p)=\min\{p,1-p\}.
\end{equation}
Optimising over the first input state $\rho$ is then equivalent to optimising $p=\Tr[\rho\,\mathsf{M_1}]$ over its achievable interval
$[\lambda_{\min}(\mathsf M_1),\lambda_{\max}(\mathsf M_1)]$ ($\lambda_{\min}(\mathsf M_1)$ and $\lambda_{\max}(\mathsf M_1)$ are the eigenvalues of $M_1$) and it gives 
\begin{align}
p_e^\star
&=\min_{\rho}\min\{\Tr[\rho\,\mathsf M_1],1-\Tr[\rho\,\mathsf M_1]\}\nonumber\\
&=\frac12\left(1-\max_{\rho}\big|2\Tr[\rho\,\mathsf M_1]-1\big|\right)\nonumber\\
&=\frac12\left(1-\|2\mathsf M_1-\mathbbm{1}\|_\infty\right)
=\frac12\left(1-\|\mathsf M_1-\mathsf M_2\|_\infty\right).
\end{align}

\noindent We now show that the value above is optimal among all two-use adaptive strategies
i.e., over all admissible triples $(p,q_1,q_2)$.

\noindent For equal priors, the Bayes optimal error can be written as
\begin{eqnarray}\label{eq:pe-TV-form}
p_e(p,q_1,q_2)&=&\frac12\sum_{i,j}\min\{P_{\mathsf M}(i,j),P_{\mathsf N}(i,j)\}\\
&=&\frac12\Big(1-\mathrm{TV}(P_{\mathsf M},P_{\mathsf N})\Big),
\end{eqnarray}
where $\mathrm{TV}(P_{\mathsf M},P_{\mathsf N}):=\tfrac12\sum_{i,j}|P_{\mathsf M}(i,j)-P_{\mathsf N}(i,j)|$
is the total variation distance between the two classical distributions.

Using the expressions for $P_{\mathsf M}(i,j)$ and $P_{\mathsf N}(i,j)$ in terms of $(p,q_1,q_2)$,
 we find
\begin{align}
&\sum_{i,j}|P_{\mathsf M}(i,j)-P_{\mathsf N}(i,j)|= \nonumber \\
&\Big(|q_1-(1-p)|+|q_1-p|\Big)+\Big(|q_2-(1-p)|+|q_2-p|\Big)\nonumber\\
&=2\max\Big\{|q_1-\tfrac12|,\ |p-\tfrac12|\Big\}
  +2\max\Big\{|q_2-\tfrac12|,\ |p-\tfrac12|\Big\},\label{eq:TV-identity} \nonumber
\end{align}
where in the second line we used the identity $|u+v|+|u-v|=2\max\{|u|,|v|\}$. Therefore,
\begin{eqnarray}
    \label{eq:TV-2use}
&\mathrm{TV}&(P_{\mathsf M},P_{\mathsf N})=
\frac12\Big(\max\{|2p-1|,|2q_1-1|\} \nonumber\\
&+&\max\{|2p-1|,|2q_2-1|\}\Big).
\end{eqnarray}

Since $p,q_1,q_2$ are all of the form $\Tr[\sigma\,\mathsf M_1]$ for some density operator $\sigma$,
they lie in the interval $[\lambda_{\min}(\mathsf M_1),\lambda_{\max}(\mathsf M_1)]$, and hence
\begin{eqnarray}\label{eq:bstar-def}
&&|2p-1|,\ |2q_1-1|,\ |2q_2-1| \le b_\star,
\\
&&b_\star:=\max\{|2\lambda_{\max}(\mathsf M_1)-1|,\ |2\lambda_{\min}(\mathsf M_1)-1|\} \nonumber \\
&&=\|2\mathsf{M_1}-\mathbbm{1}\|_\infty. \nonumber
\end{eqnarray}
Combining \eqref{eq:pe-TV-form}--\eqref{eq:bstar-def} yields the universal lower bound
\begin{equation}\label{eq:pe-lowerbound}
p_e(p,q_1,q_2)\ge \frac12(1-b_\star)=\frac12\Big(1-\|2\mathsf M_1-\mathbbm{1}\|_\infty\Big),
\end{equation}
valid for every two-use adaptive strategy.
On the other hand, choosing an eigenstate $\rho=\ket{\psi_\star}\!\bra{\psi_\star}$ of $\mathsf M_1$
whose eigenvalue $\lambda_\star$ is farthest from $1/2$ gives $|2p-1|=b_\star$; repeating the same state
in the second use (i.e., $q_1=q_2=p$) then attains $\mathrm{TV}(P_{\mathsf M},P_{\mathsf N})=b_\star$
and saturates \eqref{eq:pe-lowerbound}. Hence the bound is tight and the expression
\begin{equation}
p_{e,\star}^{(2)}=\frac12\Big(1-\|2\mathsf M_1-\mathbbm{1}\|_\infty\Big)
=\frac12\Big(1-\|\mathsf M_1-\mathsf M_2\|_\infty\Big),
\end{equation}
is optimum for two uses.

Therefore, two adaptive uses do not improve the minimum-error probability for binary labeling:
the optimum equals the single-use value and is achieved by an eigenstate of $\mathsf{M_1}$
corresponding to the eigenvalue farthest from $1/2$.

Now let us to see what happens if we can uses the measurement devices three times. After three independent uses, the outcome string $(i,j,k)\in\{1,2\}^3$ is classical, and for equal priors, the error is
\begin{equation}\label{eq:pe-classical-3}
p_e^{(3)}(p)=\frac12\sum_{i,j,k\in\{1,2\}}
\min\{P_{\mathsf M}(i,j,k),P_{\mathsf N}(i,j,k)\}.
\end{equation}
We treat the cases $p\ge \tfrac12$ and $p\le \tfrac12$ separately because the labeling
depends on the sign of $p-\tfrac12$. Indeed, under $\mathsf{M}$ the probability of outcome $1$ is $p$,
while under $\mathsf N$ it is $1-p$. Hence if $p>\tfrac12$ then outcome $1$ is more likely under
$\mathsf M$ than under $\mathsf N$, so the labeling test decides $\mathsf{M}$ when the observed string
contains a majority of $1$'s. If $p<\tfrac12$ the situation is reversed and the decision rule
flips (equivalently, it decides $\mathsf{M}$ when the number of $1$'s is small). At $p=\tfrac12$ the
two hypotheses induce identical distributions, and no strategy can outperform random guessing.

First, assume that $p\ge \tfrac12$. Then the labeling coincides with the majority vote: decide $\mathsf{M}$ iff the string contains at least two symbols
``1'' (equivalently, iff the number of ones $k$ satisfies $k\ge 2$). Hence, the error conditioned on
$\mathsf{M}$ equals the probability that the majority vote fails, i.e.\ that at most one ``1'' occurs:
\begin{align}
p_e^{(3)}(p)
&=\Pr_{\mathsf M}(k\le 1)\nonumber\\
&=\binom{3}{0}(1-p)^3+\binom{3}{1}p(1-p)^2 \nonumber\\
&=(1-p)^3+3p(1-p)^2 \label{eq:pe-3use}\\
&=(1-p)^2(1+2p).\nonumber
\end{align}
For $p\le\tfrac12$ one analogously obtains $p_e^{(3)}(p)=p^2(3-2p)$, and therefore
\begin{equation}\label{eq:pe-3use-symmetric}
p_e^{(3)}(p)=\min\big\{(1-p)^2(1+2p),\; p^2(3-2p)\big\}.
\end{equation}
Optimising over the probe state $\rho$ amounts to optimising $p=\Tr[\rho\,\mathsf M_1]$ over the interval
$[\lambda_{\min}(\mathsf M_1),\lambda_{\max}(\mathsf M_1)]$. The minimum is attained at an eigenstate of
$\mathsf M_1$ whose eigenvalue $\lambda_\star$ is farthest from $1/2$, i.e.\ such that
$b_\star:=|2\lambda_\star-1|=\|2\mathsf M_1-\mathbbm{1}\|_\infty=\|\mathsf M_1-\mathsf M_2\|_\infty$.
Assuming $\lambda_\star\ge\tfrac12$ (otherwise swap $p\leftrightarrow 1-p$), the optimal three-use error is
\begin{equation}\label{eq:pe-3use-opt}
p_{e,\star}^{(3)}=(1-\lambda_\star)^2(1+2\lambda_\star)
=\frac{(1-b_\star)^2(2+b_\star)}{4}.
\end{equation}
Comparing with the two-shot case, one can show
\begin{equation}
    p_{e,\star}^{(2)} - p_{e,\star}^{(3)}=\frac{1}{4}b_{\star}(1-b_{\star}^2)\geq 0, 
\end{equation}
which means that, unless the two-shot case, the three-shot case can reduce the minimum error of labeling. Therefore, for the case of minimum-error labeling, we have the following theorem regarding the difference between even and odd numbers in multi-shot labeling of binary measurements. 
\begin{theorem} \label{thm:even-odd-nobin}
Consider a binary measurement $\{\mathsf M_1,\mathsf M_2\}$ and its swapped-labeling version
$\{\mathsf N_1,\mathsf N_2\}$ defined by $\mathsf N_1=\mathsf M_2$ and $\mathsf N_2=\mathsf M_1$.
Fix a probe state $\rho$ and set
\begin{equation}
p:=\Pr_{\mathsf M}(1|\rho)=\Tr[\rho\,\mathsf M_1],
\end{equation}
Repeat the same probe state $n$ times and let $k$ denote the number of outcomes labeled ``1''.
Then, under $\mathsf M$ and $\mathsf N$ respectively, one has
\begin{align}\label{eq:pm-pn-k}
&\Pr_{\mathsf M}(k=t)=\binom{n}{t}p^t(1-p)^{n-t},\\
&\Pr_{\mathsf N}(k=t)=\binom{n}{t}(1-p)^t p^{\,n-t}, \nonumber
\qquad t=0,1,\ldots,n.
\end{align}
For equal priors, the optimal decision rule is the MAP test \footnote{MAP denotes the maximum-a-posteriori rule: choose the hypothesis with the larger posterior probability,
equivalently, for equal priors, the one with the larger likelihood.)}, which is a majority test on $k$.
Assuming $p\ge \tfrac12$ (the case $p\le\tfrac12$ is symmetric), the corresponding minimum-error
probabilities are
\begin{align}
\label{eq:nobin-pe-odd}
p_e^{(2m-1)}(p)
&=\Pr_{\mathsf M}\!\big(k\le m-1\big)\\
&=\sum_{t=0}^{m-1}\binom{2m-1}{t}p^t(1-p)^{2m-1-t}, \nonumber
\end{align}
\begin{align}
\label{eq:nobin-pe-even}
&p_e^{(2m)}(p)
=\Pr_{\mathsf M}\!\big(k\le m-1\big)+\frac12\Pr_{\mathsf M}(k=m)\\
&=\sum_{t=0}^{m-1}\binom{2m}{t}p^t(1-p)^{2m-t}
+\frac12\binom{2m}{m}p^m(1-p)^m. \nonumber 
\end{align}
Moreover, for every $m\ge 1$ one has
\begin{equation}\label{eq:nobin-even-odd-eq}
p_e^{(2m)}(p)=p_e^{(2m-1)}(p),
\end{equation}
i.e.\ adding one extra use to pass from an odd number of uses to the next even number does not decrease
the optimal error. Consequently, improvements (if any) occur only when passing from an even number of uses
to an odd number of uses.
\end{theorem}

\begin{proof}
The expressions \eqref{eq:nobin-pe-odd}--\eqref{eq:nobin-pe-even} follow from the majority-test description
of the MAP rule together with \eqref{eq:pm-pn-k}.
To prove \eqref{eq:nobin-even-odd-eq}, use Pascal's identity
$\binom{2m}{t}=\binom{2m-1}{t}+\binom{2m-1}{t-1}$ and reindexing to obtain
\begin{align}
&\sum_{t=0}^{m-1}\binom{2m}{t}p^t(1-p)^{2m-t}\\
&=(1-p)\sum_{t=0}^{m-1}\binom{2m-1}{t}p^t(1-p)^{2m-1-t}\nonumber\\
&+p\sum_{t=0}^{m-2}\binom{2m-1}{t}p^t(1-p)^{2m-1-t}. \nonumber\label{eq:nobin-pascal-step}
\end{align}
Substituting \eqref{eq:nobin-pascal-step} into \eqref{eq:nobin-pe-even} and noting that
$\binom{2m}{m}=2\binom{2m-1}{m-1}$ yields
\begin{align}
p_e^{(2m)}(p)
&=(1-p)\sum_{t=0}^{m-1}\binom{2m-1}{t}p^t(1-p)^{2m-1-t}\nonumber\\
&+p\sum_{t=0}^{m-2}\binom{2m-1}{t}p^t(1-p)^{2m-1-t}\nonumber\\
&\quad +\binom{2m-1}{m-1}p^m(1-p)^m\nonumber\\
&=\sum_{t=0}^{m-2}\binom{2m-1}{t}p^t(1-p)^{2m-1-t}\nonumber\\
&+\binom{2m-1}{m-1}p^{m-1}(1-p)^m\big((1-p)+p\big)\nonumber\\
&=\sum_{t=0}^{m-1}\binom{2m-1}{t}p^t(1-p)^{2m-1-t}\nonumber \\
&=p_e^{(2m-1)}(p),
\end{align}
which completes the proof.
\end{proof}

\section{Non-binary observables} \label{sec:4non-binary}
\noindent  The importance of multi-shot labeling becomes particularly evident for non-binary measurements, for which complete labeling is impossible in a single shot, since each use yields only one labeled outcome. Consequently, for an observable with \( n \) effects, at least \( n-1 \) uses of the measurement device are required for complete labeling. In general, labeling a non-binary observable composed of distinct effects proves more intricate than the binary case. This complexity stems from the fact that labeling such an observable corresponds to a distinguishability task among up to \( n! \) possible measurement channels. In this section, we study the problem of labeling for non-binary measurements. We begin by considering perfect labeling and then address minimum-error labeling for those cases where perfect labeling is not possible.

\subsection{Perfect labeling}
\noindent As established in Ref.~\cite{sudarsanan2024single}, perfect labeling of binary measurements in a single shot is possible if and only if at least one of the effects is rank-deficient. Here, we extend this analysis to non-binary, finite-outcome measurements over a finite number \( m \) of device uses. We assume throughout that all measurement effects are distinct (i.e., \( M_i \neq M_j \) for all \( i \neq j \)), unless otherwise stated. We begin with the following theorem. 
 
 \begin{theorem}
\label{theorem:non-binary-perfect labeling}
                A $d$-dimensional observable $ \mathsf{M} $ with $n \in [2,d]$ different effects can be perfectly labeled, using the simple scheme, in $(n-1)$ shots if and only if there exists at least $(n-1)$ effects having at least one eigenvalue 1. 
                \label{nonbinary perfect}
\end{theorem}
\begin{proof}
 Perfect labeling for an effect $M_i$ means the existence of a state such that $M_i\ket{\varphi_i}\neq 0$ and $M_j\ket{\varphi_i}=0$ for all $j\neq i$ which result in $\sum_{j\neq i}M_j\ket{\varphi_i}=0$. Using $\sum_j M_j=\mathbbm{1}$, we have
\begin{align}
     \sum_jM_j &\ket{\varphi_i}=\ket{\varphi_i}\nonumber \\ &\rightarrow M_i\ket{\varphi_i} +\sum_{j\neq i}M_j\ket{\varphi_i}= M_i\ket{\varphi_i}=\ket{\varphi_i}, 
\end{align}
Then for $M_i$ to be labeled perfectly, $\ket{\varphi_i}$ has to be an eigenvector with eigenvalue $1$. For a non-binary measurement, to be perfectly labeled, we need $n-1$ effects to be labeled. Therefore at least $n-1$ of effects should have at least one eigenvalue $1$ in their support. The converse is straightforward, if an effect has an eigenvector $\ket{\psi_i}$ with eigenvalue 1, i.e. then because of the resolution of the identity ($\sum_iM_i=\mathbbm{1}$), $\ket{\psi_i}$ is in the kernel of all other effects ($M_{j\neq i}\ket{\psi_i}=0$). Then $\ket{\psi_i}$ can be used to perfectly label $M_i$ and if $n-1$ effects have eigenvalue $1$ in their support then the measurement can be fully labeled. 
\end{proof}
Moreover, it is worth mentioning that for a non-binary measurement with $n$ effects, if the measurement cannot be perfectly labeled with $n-1$ uses, then it can not be improved by increasing the number of uses. In other words, $n-1$ is the fixed number for the possibility of perfect labeling. 
\begin{corollary}
If a non-binary measurement can be perfectly labeled, then each arbitrary binarization of it (each binary measurement of the form $\{M_i,\sum_{j\neq i}M_j\}$) can be perfectly labeled in a single shot.
\end{corollary}
\begin{proof}
From Theorem 1 in \cite{sudarsanan2024single}, each binary measurement can be perfectly labeled in one shot if and only if at least one of the effects is rank deficient. For a non-binary measurement with $n$ effects, there are $n$ different binarizations. According to Theorem \ref{nonbinary perfect}, at least $n-1$ effects should have at least one eigenvalue equal to $1$. Thus, in any binarization, at least one of the effects is rank deficient, and consequently, they are perfectly labeled in one shot.
\end{proof}
Utilizing Theorem \ref{nonbinary perfect}, we can characterize the measurements that can be perfectly labeled within a finite number of uses.
\begin{corollary}
    Any finite-shot perfectly labeled non-binary measurement with $n$ effects and on a $d$ dimensional space ($d \geq n$) can be parametrized as 
    \begin{eqnarray}
        M_i=\ketbra{e_i}{e_i}+A_i, \enspace \forall i \in [1, n-1] \enspace \text{and} \enspace M_n=A_n, 
    \end{eqnarray}
    where $\{\ket{e_i}\}_{i=1}^{n-1}$ is a an orthonormal basis in a $n-1$ dimensional subspace ($\mathcal{H}_{n-1}$) and $A_i$'s are semi positive operators such that $\sum_{i=1}^{n}A_i=\mathbbm{1}_{\mathcal{H}_d/ \mathcal{H}_{n-1}}$.
    \end{corollary}

\subsection{Minimum-error labeling}
In this section, we investigate the minimum-error labeling of non-binary measurements for cases where perfect labeling is not possible.
Our approach to non-binary observables begins by examining the minimum-error labeling scenario without entanglement assistance, and then we assess whether the more complex entanglement-assisted scenario can improve the performance.

Before proceeding, we note that fully labeling a non-binary measurement is typically a multi-shot task.
Labeling an $n$-outcome measurement can be formulated as identifying an unknown permutation of the outcome
labels, i.e., discriminating among $n!$ possible measure-and-prepare channels corresponding to all relabelings.
Each use of the device produces a single classical outcome and therefore provides (in general, probabilistic)
information about the underlying permutation; after each observation one may update the set of plausible
permutations (equivalently, update a posterior distribution over permutations).

In the step-by-step protocol considered below, we aim to fix the labels sequentially: in step $x$ we
choose an input state designed to maximise the probability of a particular effect $M_x$, and we declare
that the observed device label corresponds to $\mathsf{M_x}$. Conditional on the event that step $x$ is correct,
the number of permutations consistent with the information gathered so far decreases from $n!$ to $(n-1)!$,
then to $(n-2)!$, and so on. After $m$ successful steps, there remain $(n-m)!$ permutations compatible with the
identified correspondences, and after $n-1$ successful steps, the labeling is uniquely determined.

Accordingly, the probability of successful full labeling for a sequential protocol is computed via the
chain rule as
\begin{equation}
P(\text{full success})=\prod_{x=1}^{n-1} P(S_x\,|\,S_1,\cdots, S_{x-1}),
\label{eq:fullsuccess}
\end{equation}
where $S_x$ denotes the event that step $x$ identifies the correct correspondence

In the following, we study the minimum-error labeling procedure step by step. We should note that if the measurement effects include projectors (i.e., effects with eigenvectors corresponding to eigenvalue one), their corresponding outcomes can be assigned perfectly. This is achieved simply by preparing the measurement with an input state that is an eigenvector of the projector. As a result, the labeling task of a measurement with $m$ effects with eigenvalue 1 can be recast to discrimination among $(n-m)!$ prepare-and-measure channels.

The labeling process begins with the first outcome, following the simple parallel scheme illustrated in Fig. \ref{fig:testers}. To compute the minimum-error probability within this scheme, we apply a lemma established in Ref. \cite{ghoreishi2021minimum}. This lemma furnishes the minimum-error probability for distinguishing a collection of pairwise commuting states.
		\begin{lemma} [Theorem 1 in \cite{ghoreishi2021minimum}]
			Let $\{ \varrho_1,\dots, \varrho_n\}$ be a set of pairwise commuting states appearing with respective \textit{a priori} probabilities $\{ p_1, \dots, p_n\}$. Hence, we have the spectral decomposition $\varrho_j = \sum_k \lambda_{jk} \Pi_k$ for all $j$ with $\lambda_{jk} = {\rm Tr}[\varrho_j \Pi_k]$. Then, the minimum-error probability $p_\text{e}$ is given by $p_\text{e}=1-p_\text{s}$ where
	\begin{equation}
		p_{\text{s}} = \sum_k {\rm Tr}\{\Pi_k\} \max\{ {\rm Tr}\{p_1 \varrho_1 \Pi_k\},\dots,{\rm Tr}\{p_n \varrho_n \Pi_k\}\}.\label{lemma}
	\end{equation}
 Notice that the above minimum-error discrimination success probability is always achievable by the measurements composed of the projectors $\{\Pi_k\}_k$.
\end{lemma} 
\noindent Using this lemma, we arrive at the following theorem: 
	\begin{theorem} 
		A non-binary observable $\G $,  composed of  $n$ effects $ \{M_1, \dots, M_n \}$, can be partially labeled with the minimum-error probability  being $p_{\text{e}} = 1- \frac{1}{(n-1)!} \alpha_{\G}$, without the assistance of entanglement. 
	\end{theorem}
 \begin{proof}
		 We first consider the case without entanglement assistance, and then prove that entanglement does not give any advantage. For the case of without entanglement, the labeling task translates to the simultaneous discrimination of the $n!$ prepared states $\varrho_j := \mathcal{P}_{\pi_{j}}(\varrho)  = \sum_m \tra{\G_{j}(m) \varrho} \dyad{m}$, when the measured state is $\varrho$. We note that this is a collection of mutually commuting states. With respect to the above lemma, we identify $\{\Pi_m = \dyad{m}\}_{m=1}^n$\}. For our case in which each observable appears with equal chance, $p_j = \frac{1}{n!}$, Eq (\ref{lemma}) becomes
		\begin{equation}
			\begin{aligned}
				p_\text{s} =& \frac{1}{n!} \{ \max\{\tra{\varrho M_1},\tra{\varrho M_2}, \dots, \tra{\varrho M_n}\} + \\
				&\max\{\tra{\varrho M_2},\tra{\varrho M_3}, \dots, \tra{\varrho M_1}\} + \dots+ \\
				& \max\{\tra{\varrho M_n},\tra{\varrho M_1}, \dots, \tra{\varrho M_{n-1}}\} \}\\
				=& \frac{1}{n!} n \cdot \max\{\tra{\varrho M_1},\tra{\varrho M_2}, \dots, \tra{\varrho M_n}\} \\
                =& \frac{1}{(n-1)!}  \max_{M_j}\{\tra{\varrho M_j} \}.  
			\end{aligned}
		\end{equation}
     Now, optimising over the possible states $\varrho$, and denoting $\alpha_{\G} := \max_{\{\varrho, M_j\}} \text{Tr}\{\varrho M_j\} = \max_j \norm{M_j}_{\infty}$, we arrive at the optimal success probability as the following,
     \begin{equation}
         p_s^{\text{opt}} = \frac{1}{(n-1)!}  \alpha_{\G}.\label{no-ent}
     \end{equation}

Now we show that entanglemt does not provide any improvement over the simplest case. To see this, we prove for a three-effect observable, $\M = \{ M_1, M_2, M_3\}$. Assuming there are no identical effects, there are six non-identical observables and their corresponding Choi-Jamio\l kowski operators $\Phi_{k}$ which we need to discriminate in order to label. As such, the quantum testers that testing these observables are given by $\{ T_c =  \sum_{i} H_i^{(c)} \otimes \dyad{i}  \}$ with $c \in \{1, \dots, 6 \}$ and $\sum_c H_j^{(c)} = \varrho $ for all $j$. Then, the average probability of success $p_s$ is given by $p_s = \sum_j p_j \text{Tr}\{T_j \Phi_j\}$. Plugging in the expressions for ${T}_j$ and $\Phi_j$ and using the notation $H_x^{\alpha\beta} = H_x^{(\alpha)}+H_x^{(\beta)}$, we have,
\begin{eqnarray}
    p_s &=&  \frac{1}{3!} \tra{M_1 H_1^{12} +  M_2 H_1^{34} + M_3 H_1^{56}} \nonumber \\
     & & + \frac{1}{3!} \tra{M_1 H_2^{35} + M_2 H_2^{16} + M_3 H_2^{24}} \nonumber \\
     & & + \frac{1}{3!} \tra{M_1 H_3^{46}+ M_2 H_3^{25} + M_3 H_3^{13}} \nonumber  .
 \end{eqnarray}
Consider the term appearing in the first summand as $ \Lambda = \tra{M_1 H_1^{12} +  M_2 H_1^{34} +  M_3 H_1^{56}}$. Knowing that $H_1^{12} + H_1^{34} + H_1^{56} = \varrho,$ without loss of generality, we can assume $H_1^{12} = w_1 \sigma_1,  H_1^{34} = w_2 \sigma_2, \text{and }  H_1^{56} = w_3 \sigma_3$,  where $w_i$'s are probabilities adding up to one and $\sigma_i$'s are some density operators.  Now we have, 
\begin{align}
    \Lambda &= \sum_{i=1}^3 w_i \tra{M_i \sigma_i} 
     \le \sum_{i=1}^3 w_i ||M_i||_{\infty} \\
     &\le \max_i ||M_i||_{\infty} = \alpha. \nonumber
\end{align}
Thus, we find that the first summand $\frac{1}{3!}\Lambda$ is upper bounded by $ \frac{1}{3!}\alpha$. Similar analyses can be carried out for the second and third summand, leading to similar conclusions that both of them are separately upper bounded by $ \frac{1}{3!}\alpha$. Thus, we have the following inequality
\begin{equation}
    p_s \le \frac{1}{3!}(\alpha + \alpha + \alpha) = \frac{1}{3!} 3 \alpha = \frac{1}{2!} \alpha.
\end{equation}
Comparing this inequality and Eq\eqref{no-ent}, with $n=3$, we can conclude that the average success probability furnished by Eq\eqref{no-ent} is fact the optimal one over all possible strategies, including entanglement-assisted ones. This analysis can be extended to non-binary observables composed of finite $n$ effects, arriving at the conclusion that $p_s = \frac{1}{(n-1)!} \alpha_{\M}$ is the optimal average success probability. 
\end{proof}    

In brief, for a non-binary observable $\mathsf{M}$, composed of $n$ effects  $ \{M_1, \dots, M_n \}$, the optimal average success probability obtained over no-entanglement strategies, $p_{\text{s}} = \frac{1}{(n-1)!} \alpha_{\mathsf{M}}$, is optimal over all strategies. This implies that a more general entanglement-assisted strategy cannot improve it.

A protocol for achieving the optimal success probability, in Eq(\ref{no-ent}), is the following. We note that $\alpha_{\mathsf{M}}$ is the largest eigenvalue among the eigenvalues of all of the involved effects. Let $M$ be the effect having this eigenvalue. That is, there exists a state vector $\ket{\varphi}$ such that $\alpha_{\mathsf{M}} = \braket{\varphi}{M| \varphi}$. Now, the unlabeled measurement device is made to measure the state $\ket{\varphi}$ and an outcome label is registered. The outcome label associated with the effect $M$ occurs with  probability  $\alpha_{\mathsf{M}}$. Once we associate the label to  $M$, the $n!$ possibilities of observables reduce to $(n-1)!$. Now, the only possible choice is to select one of the $(n-1)!$ equally probable observables and announce it as the decision. This is done with the probability $\frac{1}{(n-1)!} \alpha_{\mathsf{M}}$, which coincides with Eq(\ref{no-ent}). Discussion of this protocol also showcases the importance of having access to multiple shots. This is due to that fact that  after use of the first (and only) shot here, we are bound to ``classically" choose one of the possible observables.   
 
Given an observable $\mathsf{M}$ with $n$ effects, let us ``re-index" and order the effects as $M_1, M_2, \dots,M_n$ according to the largest eigenvalue of each individual effect. That is, $\lambda_1 \ge \lambda_2 \ge \dots\ge \lambda_n $, where $\lambda_j$ is the largest eigenvalue of $M_j$. Now, we have the following proposition, which gives a lower bound for labeling of $m$ outcomes 
\begin{prop}\label{prop:lowerbound-simple}
Let $\{\mathsf M_j\}_{j=1}^n$ be an $n$-outcome POVM and assume the device implements an unknown
re-labeling (permutation) of the outcomes, drawn uniformly from $S_n$.
Order the effects such that $\|\mathsf M_1\|_\infty \ge \|\mathsf M_2\|_\infty \ge \cdots \ge \|\mathsf M_n\|_\infty$.
Then, for any $m\in\{1,\ldots,n-1\}$, the optimal average success probability for identifying the correct
permutation using $m$ uses satisfies the lower bound
\begin{equation}\label{eq:lowerbound-m}
p^{(m)}_{s,\mathrm{opt}} \;\ge\; \frac{1}{(n-m)!}\prod_{x=1}^{m}\|\mathsf M_x\|_\infty .
\end{equation}
\end{prop}

\begin{proof}
Consider the following sequential protocol.
At step $x=1,\ldots,m$, choose a probe state $\rho_x$ that attains
$\Tr[\rho_x \mathsf M_x]=\|\mathsf M_x\|_\infty$ (e.g.\ an eigenstate of $\mathsf M_x$ with largest eigenvalue),
use the device once, and declare that the observed device label corresponds to $\mathsf M_x$.
If at any step the observed device label coincides with a label already assigned in a previous step, the protocol
aborts and is counted as failure. After $m$ successful steps, there remain $(n-m)!$ permutations consistent with the
assigned correspondences; guessing uniformly among them yields the additional factor $1/(n-m)!$.
Consequently, the protocol succeeds with probability in Eq.\eqref{eq:lowerbound-m}, proving the bound.
\end{proof}

Let us briefly discuss the simplest non-binary setting, namely a measurement with three effects
$\{M_1,M_2,M_3\}$, and clarify the role of entanglement in multi-shot labeling. We model the unknown labeling of the measurement outcomes by a permutation
$\pi\in S_3$ acting on the true effect indices. If the physical effect is $M_j$, the device
reports the classical label $\pi(j)$. Hence, upon observing a device label $a$, the
corresponding physical effect is $M_{\pi^{-1}(a)}$, which explains the appearance of the
inverse permutation in the probabilities later in this section.

Since, the associated measurement channel is \textit{entanglement-breaking};therefore,
for any input (possibly entangled with an ancilla), the output state is classical-quantum and carries no
entanglement with any reference system. Consequently, allowing entangled probes across uses does not provide
a genuinely quantum advantage; any $m$-use entanglement-assisted tester can be reduced, without loss of
generality, to a \textit{classical adaptive protocol} in which one chooses an input state at each use based
only on the previously observed classical outcomes, followed by a classical decision rule.

For two uses (which is the first nontrivial multi-shot instance for $n=3$), a general classical-adaptive
protocol is as follows: choose an initial state $\rho^{(1)}$, observe the first device label
$a\in\{1,2,3\}$, then choose a second state $\rho^{(2)}_{a}$ depending on $a$, observe the second label
$b\in\{1,2,3\}$, and finally output an estimate $\hat\pi(a,b)\in S_3$ of the permutation.
Under hypothesis $\pi$, the single-use probability of observing label $a$ on input $\rho$ is
\[
P_{\pi}(a|\rho)=\Tr\!\big[\rho\,M_{\pi^{-1}(a)}\big],
\]
hence the induced joint distribution for the two observed labels is
\[
P_{\pi}(a,b)
=\Tr\!\big[\rho^{(1)} M_{\pi^{-1}(a)}\big]\;
 \Tr\!\big[\rho^{(2)}_{a} M_{\pi^{-1}(b)}\big].
\]
Assuming a uniform prior over $\pi\in S_3$, the optimal decoding rule is the MAP
rule, which here coincides with maximum likelihood. Therefore, the globally optimal two-use success
probability (optimised over all adaptive strategies, and hence also over all entanglement-assisted
testers) can be written as
\begin{equation}\label{eq:opt-3outcome-2use}
p_{s,\mathrm{opt}}^{(2)}
=
\max_{\rho^{(1)},\,\{\rho^{(2)}_{a}\}}
\frac{1}{6}\sum_{a=1}^3\sum_{b=1}^3
\max_{\pi\in S_3}\;
\Tr\!\big[\rho^{(1)} M_{\pi^{-1}(a)}\big]\;
\Tr\!\big[\rho^{(2)}_{a} M_{\pi^{-1}(b)}\big].
\end{equation}
Equation~\eqref{eq:opt-3outcome-2use} shows explicitly that the multi-shot optimisation reduces to choosing
(at most) one probe state per history of observed labels, followed by classical post-processing; it also
makes clear that based on Proposition~\ref{prop:lowerbound-simple}, simple scheme provides an achievable success probability for this
problem which does not by itself establish global optimality for $m>1$.
Moreover, the optimisation in \eqref{eq:opt-3outcome-2use} can be evaluated numerically (for fixed POVM
$\{M_1,M_2,M_3\}$) by formulating it as a small semidefinite program over the density operators
$\rho^{(1)}$ and $\{\rho^{(2)}_{a}\}_{a=1}^3$, together with auxiliary variables that linearise the inner
$\max_{\pi\in S_3}$.

\subsection{Case study: Trine measurement}
Consider the standard qubit trine POVM $\{M_1,M_2,M_3\}$ with
$M_k=\tfrac{2}{3}\ket{\psi_k}\!\bra{\psi_k}$, where the $\ket{\psi_k}$ are equatorial states separated by
$2\pi/3$.i.e. $\ket{\psi_0}=\ket{0}$ and $\ket{\psi_{1}}=\frac{1}{2}\ket{0}+ \frac{\sqrt{3}}{2}\ket{1}$,and $\ket{\psi_{2}}=\frac{1}{2}\ket{0}- \frac{\sqrt{3}}{2}\ket{1}$ .
The labeling process for this measurement works as follows. 
Fix the first probe state $\rho^{(1)}$ and define
\[
p_i:=\Tr[\rho^{(1)}M_i],\qquad i=1,2,3,
\]
so that $p_1+p_2+p_3=1$. The unknown labeling is a permutation $\pi\in S_3$, assumed uniform,
$\Pr(\pi)=1/6$. If the device outputs the classical label $A=a\in\{1,2,3\}$, then the underlying
true effect that has occurred is the index $I=\pi^{-1}(A)$. Under hypothesis $\pi$, the likelihood
of observing $A=a$ is $\Pr(A=a\mid \pi)=p_{\pi^{-1}(a)}$.
For any $i\in\{1,2,3\}$, the posterior distribution of $I$ conditioned on observing $A=a$ can be computed as 
\begin{eqnarray}
\label{eq:posterior-bayes}
\Pr(I=i\mid A=a)=\Pr(\pi^{-1}(a)=i\mid A=a)\\
=\frac{\Pr(A=a\mid \pi^{-1}(a)=i)\,\Pr(\pi^{-1}(a)=i)}{\Pr(A=a)}, \nonumber
\end{eqnarray}

First, if $\pi^{-1}(a)=i$ (equivalently $\pi(i)=a$), then for every such permutation $\pi$ we have
$\Pr(A=a\mid \pi)=p_{\pi^{-1}(a)}=p_i$, hence
\begin{equation}\label{eq:likelihood-cond}
\Pr(A=a\mid \pi^{-1}(a)=i)=p_i.
\end{equation}
Second, under the uniform prior on $\pi$, the event $\pi^{-1}(a)=i$ holds for exactly $(3-1)!=2$
permutations (fix $\pi(i)=a$ and permute the remaining two indices), hence
\begin{equation}\label{eq:prior-cond}
\Pr(\pi^{-1}(a)=i)=\frac{2}{6}=\frac{1}{3}.
\end{equation}
Finally, by the law of total probability and \eqref{eq:likelihood-cond}--\eqref{eq:prior-cond},
\begin{eqnarray}\label{eq:marginal-a}
\Pr(A=a)&=&\sum_{i=1}^3 \Pr(A=a\mid \pi^{-1}(a)=i)\Pr(\pi^{-1}(a)=i) \nonumber\\
&=&\sum_{i=1}^3 p_i\cdot\frac{1}{3}
=\frac{1}{3}\sum_{i=1}^3 p_i
=\frac{1}{3}.
\end{eqnarray}
Substituting \eqref{eq:likelihood-cond}, \eqref{eq:prior-cond}, and \eqref{eq:marginal-a} into
\eqref{eq:posterior-bayes} yields
\begin{equation}\label{eq:posterior-result}
\Pr(I=i\mid A=a)=\Pr(\pi^{-1}(a)=i\mid A=a)=p_i.
\end{equation}
In particular, the posterior distribution of the hidden index $I=\pi^{-1}(a)$ conditioned on observing
$A=a$ is \textit{independent of $a$}: it is always the same vector $(p_1,p_2,p_3)$ (up to the trivial
relabeling of the device symbols).

As a result of this discussion, we have the following lemma:

\begin{lemma}\label{lem:trine-no-adapt}
For the problem of labeling a trine measurement in two uses, any adaptive strategy (second probe depending on the first observed device label $a$) can be
replaced, without loss in average success probability, by a non-adaptive one with a fixed second probe
state $\rho^{(2)}$ (independent of $a$).
\end{lemma}

\noindent Consequently, for two-use optimal success probability for trine labeling, we have the following theorem
\begin{theorem}\label{thm:trine-2use-opt}
 The optimal average success probability for the labeling of the trine measurement using two uses of the measurement (equal priors) is
\begin{equation}\label{eq:trine-2use-opt}
p_{s,\mathrm{opt}}^{(2)}=\frac{3+\sqrt{3}}{9}\approx 0.52578.
\end{equation}
\end{theorem}

\begin{proof}
By Lemma~\ref{lem:trine-no-adapt}, it suffices to consider two fixed probe states $\rho^{(1)}$ and
$\rho^{(2)}$. Let
\begin{equation}
p_i:=\Tr[\rho^{(1)}M_i],\qquad q_i:=\Tr[\rho^{(2)}M_i],\qquad i=1,2,3.
\end{equation}
Under permutation $\pi$, the joint probability of observing device labels $(a,b)\in\{1,2,3\}^2$ is
\begin{equation}
P_\pi(a,b)=p_{\pi^{-1}(a)}\,q_{\pi^{-1}(b)}.
\end{equation}
With uniform prior over $\pi$, the MAP success probability is
\begin{equation}\label{eq:ps-map}
p_s^{(2)}=\frac{1}{6}\sum_{a,b\in\{1,2,3\}} \max_{\pi\in S_3} P_\pi(a,b).
\end{equation}
For $a=b$, one must assign the same hidden index to both labels, hence
$\max_{\pi}P_\pi(a,a)=\max_i p_iq_i$. For $a\neq b$, one must assign two distinct indices, hence
$\max_{\pi}P_\pi(a,b)=\max_{i\neq j} p_iq_j$. Since there are $3$ diagonal pairs and $6$ off-diagonal
pairs, \eqref{eq:ps-map} becomes
\begin{equation}\label{eq:ps-3out}
p_s^{(2)}=\frac12\,\max_i(p_iq_i)\;+\;\max_{i\neq j}(p_iq_j).
\end{equation}

By rotational symmetry we may take $\rho^{(1)}=\ket{\psi_1}\!\bra{\psi_1}$, yielding
\begin{equation}
(p_1,p_2,p_3)=\left(\frac{2}{3},\,\frac{1}{6},\,\frac{1}{6}\right).
\end{equation}
Restrict to equatorial pure states for $\rho^{(2)}$ (optimality follows from convexity), and let $\delta$
be its azimuthal angle relative to $\ket{\psi_1}$, i.e. $\ket{\varphi}=\frac{1}{\sqrt{2}}\ket{0}+\frac{1}{\sqrt{2}}e^{i\delta}\ket{1}$. Then
\begin{align}
&q_1(\delta)=\frac{1+\cos\delta}{3},\nonumber\\
&q_2(\delta)=\frac{1+\cos(\delta-\tfrac{2\pi}{3})}{3},\nonumber\\
&q_3(\delta)=\frac{1+\cos(\delta-\tfrac{4\pi}{3})}{3}. \nonumber 
\end{align}
For $\delta\in[0,\frac{2\pi}{3}]$, the maximum in \eqref{eq:ps-3out} are attained by
$\max_i(p_iq_i)=\tfrac{2}{3}q_1(\delta)$ and
$\max_{i\neq j}(p_iq_j)=\tfrac{2}{3}\max\{q_2(\delta),q_3(\delta)\}=\tfrac{2}{3}q_2(\delta)$,
hence
\begin{equation}
p_s^{(2)}(\delta)
=\frac{1}{3}q_1(\delta)+\frac{2}{3}q_2(\delta).
\end{equation}
Using $\cos(\delta-\tfrac{2\pi}{3})=-\tfrac12\cos\delta+\tfrac{\sqrt3}{2}\sin\delta$ gives
\begin{equation}
p_s^{(2)}(\delta)=\frac{1}{3}+\frac{\sqrt3}{9}\sin\delta,
\end{equation}
which is maximised at $\delta=\tfrac{\pi}{2}$ with value \eqref{eq:trine-2use-opt}. For $\delta\in[2\pi/3,\pi]$ one finds $\max_i(p_iq_i)=\tfrac16 q_2(\delta)$ and thus $p_s^{(2)}(\delta)\le \tfrac12<\tfrac{3+\sqrt3}{9}$, so the optimum occurs for $\delta\in[0,2\pi/3]$.

\end{proof}
An optimal protocol is then obtained by choosing the first probe to be one of the trine states $\ket{\psi_1}$
and the second probe to be an equatorial state rotated by $\pi/2$ with respect to $\ket{\psi_1}$.  

For comparison, consider the simple two-shot protocol underlying the lower bound in
Proposition~\ref{prop:lowerbound-simple}: in the first use we probe with $\dyad{\psi_1}$ and in the
second use with $\dyad{\psi_2}$. Since for the trine POVM each effect satisfies
$\|M_k\|_\infty=\tfrac{2}{3}$, this strategy achieves
\[
p_s^{(2)}=\frac{2}{3}\cdot\frac{2}{3}=\frac{4}{9}\approx 0.44,
\]
which is strictly smaller than the optimal value $p_{s,\mathrm{opt}}^{(2)}\approx 0.52578$ established in
Theorem~\ref{thm:trine-2use-opt}.

\section{conclusion} \label{sec:conclusion}
In this work, we extended the single-shot quantum labeling problem of Ref.~\cite{sudarsanan2024single} to the multiple-shot
regime and investigated how access to a finite many uses of an unlabeled measurement device affects the
possibility and performance of labeling. For binary observables, we showed that if perfect labeling is
impossible in a single shot (equivalently, both effects are full rank), then it remains impossible with any
finite number of shots. We then characterised the minimum-error performance in the swapped-labeling
scenario and found an even--odd effect: adding one use to go from an odd number of shots to the next even
number does not improve the optimal error, whereas improvements occur when passing from an even number
of shots to an odd number.

For non-binary observables, we provided a finite-shot characterisation of perfect labellability via a simple
scheme: a $d$-dimensional observable with $n$ outcomes can be perfectly labelled in $n-1$ shots if and only if
at least $n-1$ effects have an eigenvalue equal to $1$. We also studied minimum-error labeling in the
single-shot setting, obtaining a closed-form expression for the optimal success probability and showing
that entanglement assistance does not improve this optimum. Finally, we discussed multi-shot schemes for
partial identification of labels and illustrated the results on concrete examples, including the
qubit trine POVM, where the optimal two-shot success probability can be evaluated explicitly.


\section*{Acknowledgements}
We thank Anna Jenčová for the fruitful discussion and for commenting on the first version of the results.
SAG acknowledges the Štefan Schwarz Support Fund and the project VEGA 2/0164/25 (QUAS). NSR acknowledges the support of projects APVV-22-0570 (DEQHOST) and the financial support from the US Department of Energy, Office of Science, Advanced Scientific Computing Research program, under award number DE-SC0025430. SS acknowledges funding through PASIFIC program call 2 (Agreement No. PAN.BFB.S.BDN.460.022 with the Polish Academy of Sciences). This project has received funding from the European Union’s Horizon 2020 research and innovation programme under the Marie Skłodowska-Curie grant agreement No 847639 and from the Ministry of Education and Science of Poland. MZ acknowledges support from the QENTAPP 09103-03-V04-00777 project.

\bibliography{labeling}
\end{document}